\documentclass[11pt,a4paper]{article}

\usepackage{fullpage}
\usepackage{amsmath}
\usepackage{amssymb}
\usepackage{amsthm}
\usepackage{color}
\usepackage{cite}
\usepackage{graphicx}
\usepackage{hyperref}
\usepackage{verbatim}
\usepackage{subfigure}
\RequirePackage{lineno}

\newtheorem{theorem}{Theorem}
\newtheorem{lemma}[theorem]{Lemma}

\newtheorem{corollary}[theorem]{Corollary}

\newtheorem{remark}{Remark}
\newtheorem{open}{Open Problem}

\usepackage[paper=a4paper,left=25.4mm,right=25.4mm,top=25.4mm,bottom=30mm,bindingoffset=0mm]{geometry}

\newcommand{\Hlinecell}{{\mathcal H}_{line-cell}}
\newcommand{\Hvertexcell}{{\mathcal H}_{vertex-cell}}
\newcommand{\Hcellzone}{{\mathcal H}_{cell-zone}}
\newcommand{\Hyper}{{\mathcal H}}

\definecolor{SEB}{rgb}{0,0,1}

\newcommand{\mati}[1]{\textcolor{red}{\textsc{Mati:} #1}}

\newcommand{\sholong}[2]{#2}

\newcommand{\proofguardwithcells}{The first iterations of the algorithm select cells that each cover four new lines. 
We iteratively select a cell covering four lines as long as the average number of segments of uncovered lines bounding a cell is strictly greater than three. The total number of segments is $n^2$, and each contribute to two cells. Every selected cell discards four lines, and exactly $4\times 2n=8n$ segments of those. If the cell is bounded by more than four lines, we only discard exactly four of them, arbitrarily. The total number of cells after the $i$th iteration is $|C|-i$. The number of iterations is the largest value $i$ that satisfies:
\begin{eqnarray*}
\frac{2n^2 - 8in}{|C|-i} & > & 3, \\
\frac{2n^2 - 8in}{n(n+1)/2 +1 -i} & > & 3, \\
i & \sim & \frac n{16} +o(n) .
\end{eqnarray*}
Hence we can select roughly $\frac n{16}$ cells, covering together $\frac n4$ lines.

In the second phase of the algorithm, we iteratively select cells covering three new lines. Following the same reasoning, and taking into account the $i\simeq n/16$ previously selected cells, we know that the number $j$ of iterations satisfies:
\begin{eqnarray*}
\frac{2n^2 - 8in - 6jn}{|C| -i -j} & > & 2, \\
\frac{2n^2 - n^2/2 - 6jn}{n(n+1)/2 +1 -n/16 -j} & > & 2, \\
j & \sim & \frac n{12} +o(n).
\end{eqnarray*}
Hence we can select roughly $\frac n{12}$ more cells, covering together $\frac n4$ lines.

Overall, we now have $\frac n{16}+\frac n{12} + o(n)$ cells covering $\frac n2$ lines. It remains to cover the remaining $\frac n2$ lines with $\frac n4$ cells, each covering two lines, as in Theorem~\ref{warmup:guardwithcells}. The total number of cells is therefore
\begin{equation*}
\frac n{16}+\frac n{12} + \frac n4 + o(n) = \frac{19}{48} n +o(n).
\end{equation*}}

\title{Coloring and Guarding Arrangements\thanks{Research supported by the the ESF EUROCORES programme EuroGIGA, CRP ComPoSe: Fonds National de la Recherche Scientique (F.R.S.-FNRS) - EUROGIGA NR 13604, for Belgium, and MICINN Project EUI-EURC-2011-4306, for Spain.}}
\author{Prosenjit Bose\thanks{Carleton University, Ottawa, Canada. \tt{jit@scs.carleton.ca}}
\and Jean Cardinal\thanks{Universit\'e Libre de Bruxelles (ULB), Brussels, Belgium. \tt{\{jcardin, secollet, mkormanc,slanger, perouz.taslakian\}@ulb.ac.be}} 
\and S\'ebastien Collette\footnotemark[3]~~\thanks{Charg\'e de Recherches du F.R.S.-FNRS.}  
\and Ferran Hurtado\thanks{Universitat Polit\`{e}cnica de Catalunya (UPC), Barcelona, Spain.
 {\tt ferran.hurtado@upc.edu}. Partially supported by projects MTM2009-07242 and Gen. Cat. DGR 2009SGR1040.} 
  \and Matias Korman\footnotemark[3]
  \and Stefan Langerman\footnotemark[3]~~\thanks{Ma\^itre de Recherches du F.R.S.-FNRS.} 
  \and Perouz Taslakian\footnotemark[3]}

\begin{document}
\maketitle

\begin{abstract}
Given a simple arrangement of lines in the plane, what is the minimum number $c$ of colors required to color the lines so that no cell of the arrangement is monochromatic? In this paper we give bounds on the number c both for the above question, as well as some of its variations. We redefine these problems as geometric hypergraph coloring problems. If we define $\Hlinecell$ as the hypergraph where vertices are lines and edges represent cells of the arrangement, the answer to the above question is equal to the chromatic number of this hypergraph. We prove that this chromatic number is between $\Omega (\log n / \log\log n)$. and $O(\sqrt{n})$.

Similarly, we give bounds on the minimum size of a subset $S$ of the intersections of the lines in $\mathcal{A}$ such that every cell is bounded by at least one of the vertices in $S$. This may be seen as a problem on guarding cells with vertices when the lines act as obstacles. The problem can also be defined as the minimum vertex cover problem in the hypergraph $\Hvertexcell$, the vertices of which are the line intersections, and the hyperedges are vertices of a cell. Analogously, we consider the problem of touching the lines with a minimum subset of the cells of the arrangement, which we identify as the minimum vertex cover problem in the $\Hcellzone$ hypergraph.


\end{abstract}
\sholong{\newpage}{}

\section{Introduction}
While dual transformations may allow converting a combinatorial geometry problem about a \emph{configuration of points} into a problem about an \emph{arrangement of lines}, or reversely, the truth is that most mathematical questions appear to be much cleaner and natural in only one of the settings. In many cases, the dual version is considered solely when, besides making sense, it is additionally useful. Both kinds of geometric objects have inspired many problems and attracted much attention. For finite point sets the \emph{Erd\H{o}s-Szekeres problem} on finding large subsets in convex position, or the \emph{repeated distances problem} on how many times can a single distance appear between pairs of points, are examples of famous questions that have been pursued for decades and are still open. Many research problems of this kind are described in Chapter 8 in \cite{BMP}. Concerning arrangements of lines, possibly the most prevalent problems consist of studying the number of cells of each size, say triangles, that appear in every arrangement, but many other issues have been considered (see \cite{Fe,Gru1,Gru2}). There are also problems that combine both kinds of objects, like counting incidences between points and lines, or studying the arrangements of lines spanned by point sets, which includes the celebrated \emph{Sylvester-Gallai problem} on ordinary lines (those that only contain two of the points) \cite{BMP}.

In the first class of problems, substantial attention has been focusing on \emph{colored pointsets}, i.e., configurations of points that belong to several classes, the \emph{colors}, including chromatic variants of the repeated distances problem and the Erd\H{o}s-Szekeres problem, and colored versions of \emph{Tverberg's Theorem} and \emph{Helly's Theorem}. In particular there is a vast body of research on problems involving a set of \emph{red} points and a set of \emph{blue} points. Refer to \cite{KK} for a survey on red-blue problems, or to \cite{BMP} for a more generic account.

Somehow surprisingly, there is not a comparable set of questions that have been posed for colored arrangements of lines. There is a series of papers on the problem of taking bicolored sets of lines, calling \emph{monochromatic vertex} an intersection point contained only in lines of one of the colors, and discussing their existence and number \cite{Gru3,Gru4,Mo}. Another series of papers study the colorings of the so called arrangement graphs, in which vertices are the intersection points and edges are the segments between any two that are consecutive on one of the lines \cite{BEW,FHNS}.

However, many other natural questions can be asked. For example, is it true that every bicolored arrangement of lines has a monochromatic cell? We prove in this paper that the generic answer is no, but that it is yes when the colors are slightly unbalanced. This leads immediately to another question that we discuss in our work: How many colors are always sufficient, and occasionally necessary, to color any set of $n$ lines in such a way that the induced arrangement contains no monochromatic cell?

The last question brings manifestly the flavor of \emph{Art Gallery Problems} \cite{OR1,OR2,URR}. We also consider for line arrangements several issues on this topic that apparently have not been studied before: How many vertices of an arrangement suffice to guard all its cells? How many lines are enough to guard (touch) all the cells?

While coloring and guarding arrangements of lines may appear at first glance as unrelated problems, there is a clean unifying framework provided by considering appropriate geometric hypergraphs. For example, minimally coloring an arrangement while avoiding monochromatic cells can be reformulated as follows: Let $\Hlinecell$ be the geometric hypergraph where vertices are lines and edges represent cells of the arrangement. What is the chromatic number of the hypergraph? Here a proper coloring is one where no hyperedge is monochromatic.

In this work we consider several questions as the ones described in the preceding paragraphs, which arise as quite fundamental in terms of coloring and guarding arrangements of lines, and translate consistently into problems on geometric hypergraphs,  like size of a maximal independent set, size of a vertex cover, or some coloring parameter.

The terminology for hypergraphs on arrangements is introduced in Section \ref{sec:definitions}, where we also provide a table summarizing our results. Coloring problems are then discussed in Section \ref{sec:coloring} and guarding problems in Section \ref{sec:guarding}. We conclude with some observations and open problems.

\section{Definitions and Summary of Results}\label{sec:definitions}
Let $\mathcal{A}$ be an arrangement of a set of lines $L$ in $\mathbb{R}^2$. We say that an arrangement of lines $\mathcal{A}$ is {\em simple} if every two lines intersect, and no three lines have a common intersection point. From now on, we only consider simple arrangements of lines\footnote{For non-simple arrangements, the answer to most of the problems we study are either trivial or not well defined.}. 

Any arrangement $\mathcal{A}$ decomposes the plane into different cells, where a {\em cell} is a maximal connected component of $\mathbb{R}^2\setminus L$. We define $\Hlinecell = (L,C)$ as the geometric hypergraph corresponding to the arrangement, where $C$ is the set containing all cells defined by $L$. Similarly, $\Hvertexcell= (V,C)$ is the hypergraph defined by the vertices of the arrangements and its cells, where $V={L\choose 2}$ is the set of intersection of lines in $\mathcal{A}$. Finally, $\Hcellzone = (C, Z)$ is the hypergraph defined by the cells of the arrangement and its zones. The zone of a line $\ell$ in $\mathcal{A}$ is the set of cells bounded by $\ell$. The set $Z$ is defined as the set of subsets of $C$ induced by the zones of $\mathcal{A}$. Note that this hypergraph is the dual hypergraph of $\Hlinecell$.

An \emph{independent set} of a hypergraph $\Hyper = (V,E)$ is a set $S \subseteq V$ such that $\forall e \in E: e\not \subseteq S$. This definition is the natural extension from the graph variant, and requires that no hyperedge is completely contained in $S$. Analogously, a \emph{vertex cover} of $\Hyper$ is a set $S \subseteq V$ such that $\forall e \in E: e \cap S \neq \emptyset$. The \emph{chromatic number} $\chi(\Hyper)$ of $\Hyper$ is the minimum number of colors that can be assigned to the vertices $v \in V$ so that $\forall e \in E: \exists v_1, v_2 \in e : col(v_1) \neq col(v_2)$; that is, no hyperedge is monochromatic.

In the forthcoming sections we give upper and lower bounds on the worst-case values for these quantities on the three hypergraphs defined from a line arrangement. Our results are summarized in Table \ref{tabres}. Note that the maximum independent set and minimum vertex cover are complementary problems. As a result, any lower bound on one gives an upper bound on the other and vice versa. This property, along with the facts that $|L|=n$, $|V|={n\choose 2}$, and $|C|={{n(n+1) \over 2} +1}$, are used to complement many entries of the table. 


The definitions of an independent set and a proper coloring of the $\Hlinecell$ hypergraph of an arrangement are illustrated in Figures~\ref{fig:independent} and \ref{fig:coloring}, respectively. Similarly, the definition of a vertex cover of the $\Hvertexcell$
and $\Hcellzone$ hypergraphs are illustrated in Figures~\ref{fig:guardingcells}, and \ref{fig:guardinglines}, respectively.

\begin{table}
\begin{center}
\begin{tabular}{|c|c|c|c|}
\hline
Hypergraph & Max. Ind. Set & Vertex Cover & Chromatic number\\
\hline
$\Hlinecell$ & $\geq {\sqrt{n}\over 2}$ (Th.~\ref{th:is})& $\ge \frac{n}{3}$ (Cor. \ref{cor:CL})& $\Omega (\log n/\log\log n)$ (Th.~\ref{LB:chroma})\\
		    & $\le \frac{2n}{3}$ (Th.~\ref{th:ubis}) 
		    & $\le n-{\sqrt{n}\over 2}$ (Cor. \ref{cor:CL})& $\leq2\sqrt{n}+O(1)$  (Th.~\ref{UB:chrom})\\
\hline
$\Hvertexcell$ & $\geq{n^2 \over 3}-{5n \over 2}$ (Cor. \ref{cor:VCIS}) & $\geq {n^2 \over 6}$ (Th.~\ref{LB:guards})& 2 (Th.~\ref{tight:chromvertex})\\
 & $\leq {n^2 \over 3}-{n \over 2}$ (Cor. \ref{cor:VCIS}) &$\leq {n^2 \over 6}+2n$ (Th.~\ref{LB:guards})& \\
\hline
$\Hcellzone$ & $\ge {n^2 \over 2} + {5n \over 48} -o(n)$ (Cor. \ref{cor:CZIS})& $\ge {n\over 4}$ (Th.~\ref{guardwithcells})& 2 (Th.~\ref{tight:chromcells})\\
$					 $ & $\le {n^2 \over 2} + {5n \over 4} +1$ (Cor. \ref{cor:CZIS})& $ \le {19n \over 48}+ o(n)$ (Th.~\ref{guardwithcells})&\\
\hline
\end{tabular}
\end{center}
\caption{Summarizing table with the worst-case bounds for the different problems studied in this paper. 
}\label{tabres}
\end{table}

\begin{figure}
\begin{center}
\subfigure[\label{fig:independent}The thick lines form an independent set in the $\Hlinecell$ hypergraph: no cell is bounded by
those lines only.]{\includegraphics[width=.4\textwidth]{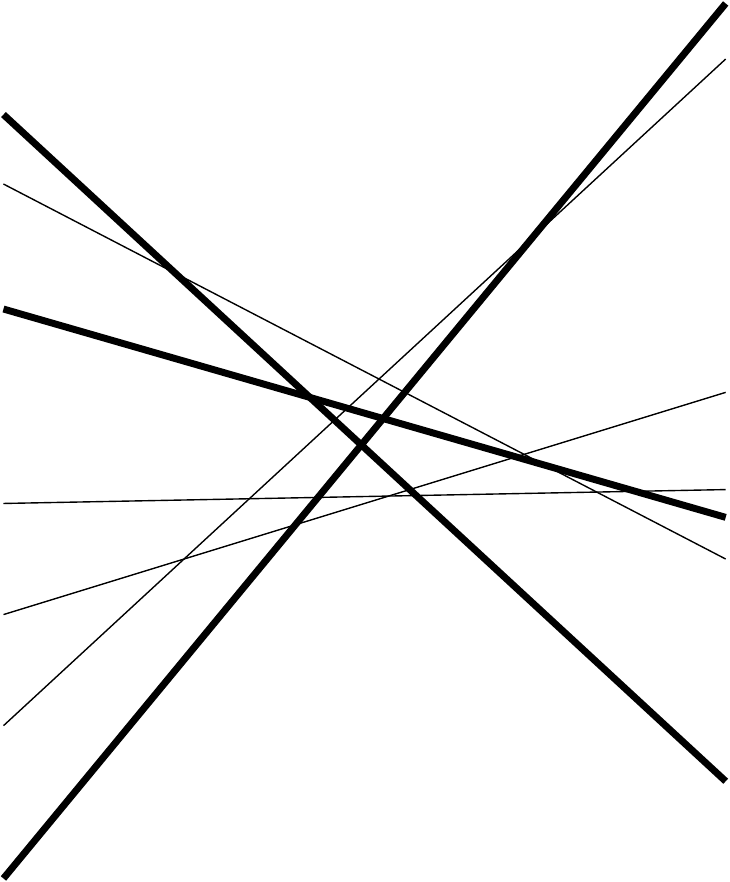}}\hspace{1cm}
\subfigure[\label{fig:coloring}A proper 3-coloring of the $\Hlinecell$ hypergraph: no cell is monochromatic.]{\includegraphics[width=.4\textwidth]{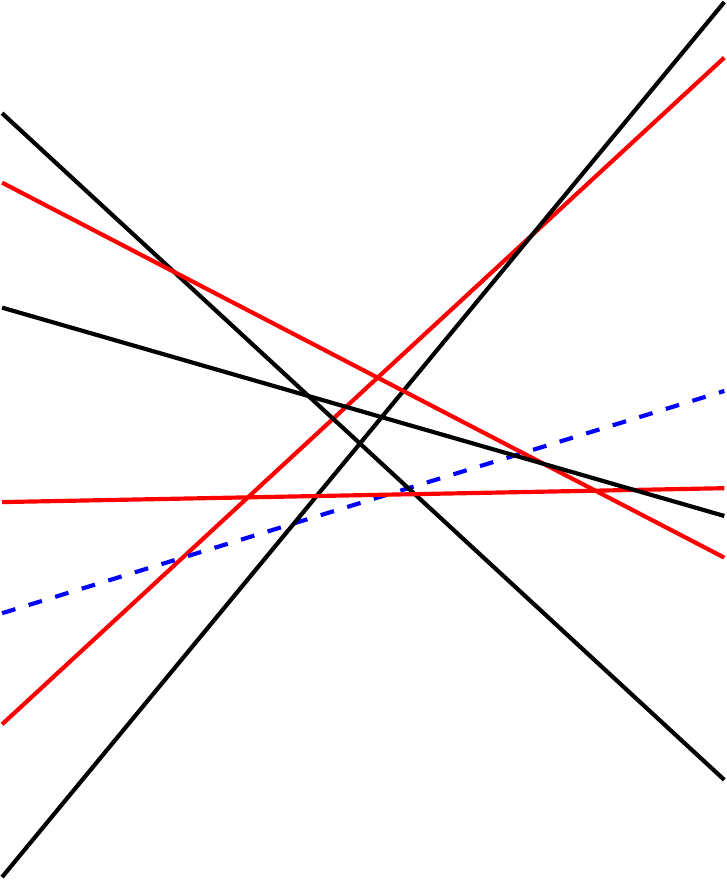}}\hspace{1cm}
\subfigure[\label{fig:guardingcells}The marked intersections form a vertex cover of the $\Hvertexcell$ hypergraph: every cell has at least
one such intersection on its boundary. That is, these vertices \emph{guard} the cells.]{\includegraphics[width=.4\textwidth]{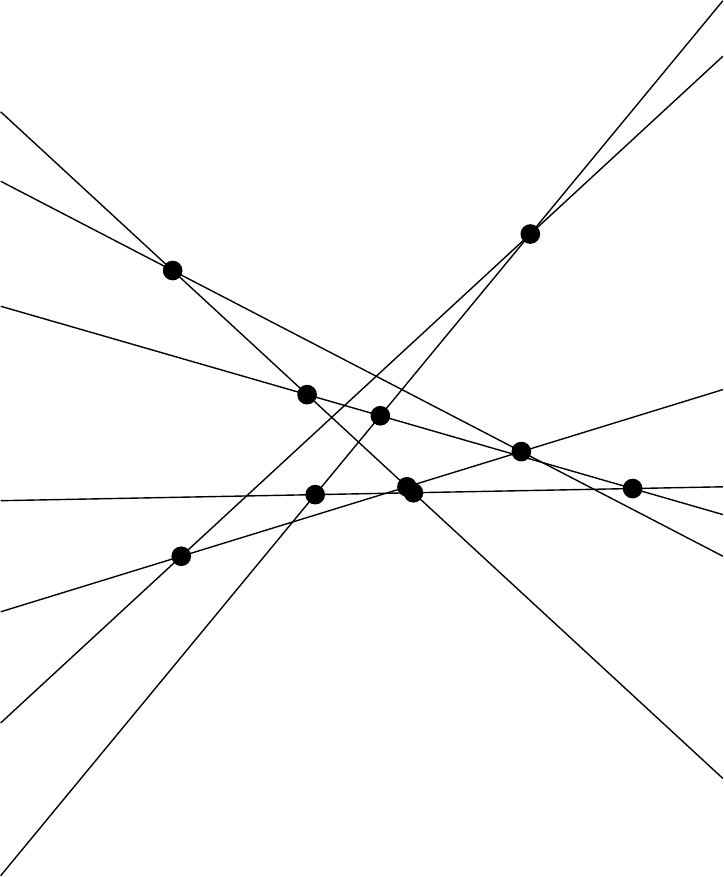}}\hspace{1cm}
\subfigure[\label{fig:guardinglines}The two marked cells form a vertex cover of the $\Hcellzone$ hypergraph: every line has a segment that lies on the boundary
of one of those cells. That is, these two cells \emph{guard} the lines.]{\includegraphics[width=.4\textwidth]{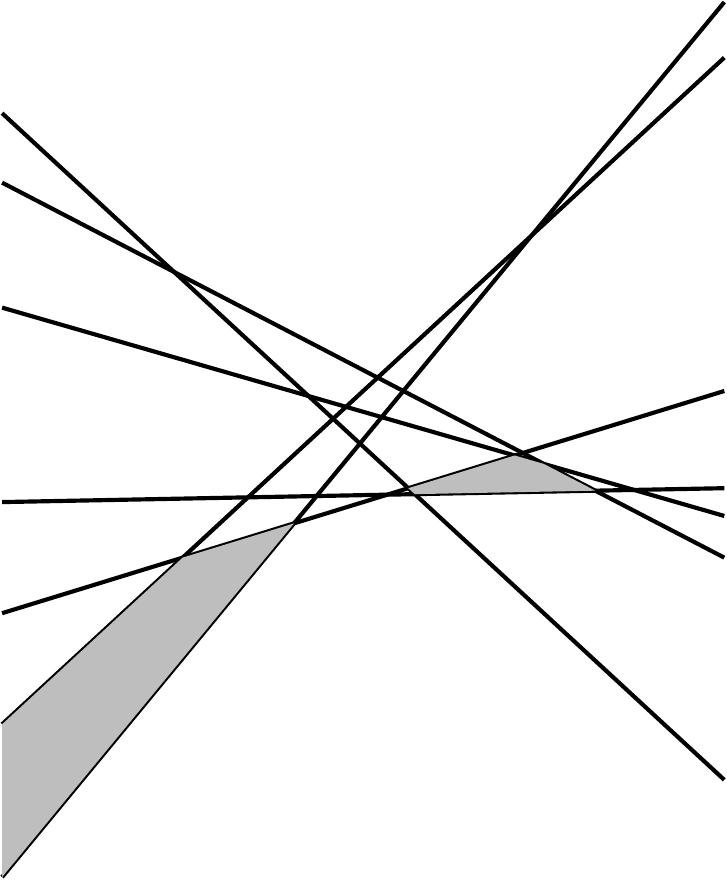}}
\end{center}
\caption{\label{fig:examples}Illustrations of the definitions.}
\end{figure}

\section{Coloring Lines, and Related Results}\label{sec:coloring}

We first consider the chromatic number of the line-cell hypergraph of an arrangement, that is, the number of colors required for coloring the lines so that no cell has a monochromatic boundary. At the end of the section we include some similar results.

\subsection{Two-colorability}

We say that a set of lines $L$ is {\em $k$-colorable} if we can color $L$ with $k$-colors such that no cell is monochromatic (in other words, the corresponding $\Hlinecell$ hyper graph has chromatic number $k$). Any coloring $c:L \rightarrow \{0,\ldots, k\}$ that satisfies such a property is said to be {\em proper}. We first tackle the (simple) question of whether the two-colorable $\Hlinecell$ hypergraphs have bounded size:

\begin{theorem}\label{prop:2colors}
There are arbitrarily large two-colorable sets of lines.
\end{theorem}


\begin{proof} An infinite family of such examples are provided by a set of $2q+1$ lines in convex position (for any $q\in \mathbb{N}$\sholong{)}{, see Figure~\ref{fig:bichrom})}.Observe that in such arrangement, each cell is either bounded by $(1)$ two consecutive lines, $(2)$ the first and the last line or $(3)$ all lines of the arrangement. It is easy to check that, if we color the lines alternatively red and blue by order of slope, no cell will be monochromatic.
\end{proof}

The coloring used in Theorem \ref{prop:2colors} uses essentially the same number of lines of each color. We prove next that the result cannot hold when the numbers are unbalanced to a small extent.

\begin{theorem}
Each color class of a proper 2-coloring $c:L\rightarrow \{0,1\}$ of a set $L$ of $n$ lines has size at most ${n\over 2} + {\sqrt{n-1}-1\over 2}$.
\end{theorem}
\begin{proof}
Let $R$ be the set of lines that are assigned color $0$, and $B$ the set of lines whose color is $1$. Let $\mathcal{A}_{R}$ denote the arrangement of the lines in $R$. As $R$ does not completely define a cell of $\mathcal A$, each cell of $\mathcal{A}_{R}$ must be traversed by a line in $B$.

We proceed iteratively: we start with $\mathcal{A}_{R}$, and add the lines in $B$ one at a time. When adding a line $\ell$, some cells of $\mathcal{A}_{R}$ will be split into two parts by a line segment induced by $\ell$. A connected component of segments inside a cell of $\mathcal{A}_{R}$ is a set of segments whose intersection graph is connected.

To each cell $c$ of $\mathcal{A}_{R}$, we assign a number, representing the number of connected components defined by the segments inside $c$. Let $f(c,i)$ denote the number of connected components inside the cell $c$ after adding the $i-{th}$ line of $B$, and $f(\mathcal{A}_{R},i)$ is the sum of $f(c,i)$ over all cells $c \in \mathcal{A}_{R}$.

When the first line $\ell_{1} \in B$ is added, $f(c,1)=1$ for each cell $c$ crossed by $\ell_{1}$, and remains zero for every other cell. Therefore, $f(\mathcal{A}_{R},1)=|R|+1$. In general, when the $i-{th}$ line $\ell_{i}$ is added, $f(\mathcal {A}_{R},i)$ increases by $|R|+2-i$. Indeed, in each cell $c$, the blue line can only intersect each component once, otherwise the corresponding segments would create a cycle, meaning a new face bounded only by blue lines, and thus monochromatic. This implies that $\ell_{i}$ intersects all previous $i-1$ lines in $i-1$ disjoint components. Inside a cell $c$, if a line $\ell_{i}$ intersects $t$ components, then $f(c,i)= f(c,i-1)-t +1$. Thus, $f(\mathcal {A}_{R},i)= f(\mathcal {A}_{R},i-1)- (i-1) + |R|+1=  f(\mathcal {A}_{R},i-1)+ |R|+2 -i$.

What we also know, is that at the end of the process, each cell of $\mathcal {A}_{R}$ should contain at least one component, otherwise the cell is monochromatic. Thus $f(\mathcal {A}_{R},|B|)$ should be bigger or equal to the number of cells in $\mathcal {A}_{R}$.

We get:
\begin{eqnarray*}
f(\mathcal {A}_{R},|B|) & = & \sum_{i=1}^{|B|} |R|+2-i,\\
\sholong{}{{|R| \cdot (|R|+1) \over 2}+1 & \le & \sum_{i=1}^{|B|} |R|+2-i,\\}
{|R| \cdot (|R|+1) \over 2} +1& \le & |B| \cdot (|R|+2) - \sum_{i=1}^{|B|} i,\\
\sholong{}{{|R| \cdot (|R|+1) \over 2} +1& \le & |B| \cdot (|R|+2) - {|B|\cdot |B|+1 \over 2},\\}
{|R| \cdot (|R|+1) \over 2} +1& \le & (n-|R|) \cdot (|R|+2) - {(n-|R|)\cdot (n-|R|)+1 \over 2},\\
|R| & \le & {n+ \sqrt{n-1}-1\over 2}.
\end{eqnarray*}

\noindent which concludes the proof.
\end{proof}

\subsection{Independent lines in $\Hlinecell$}\label{sec_indep}

Recall that an independent set of lines in an arrangement is defined as a subset of lines $S$ so that no cell of the arrangement is only adjacent to lines in $S$.

\begin{theorem}\label{th:is}
For any set $L$ of $n$ lines, the corresponding $\Hlinecell$ hypergraph has an independent set of size $\sqrt{n}/2$.
\end{theorem}
\begin{proof}
We prove that any (inclusionwise) maximal independent set has size $\Omega (\sqrt{n})$. Consider such a maximal independent set $I\subset L$ of size $i$. 
By maximality, each line $\ell\in L\setminus I$ can be associated with at least one cell of $L$ whose boundary consists only of one segment of $\ell$, and segments of  lines in $I$. We choose one such cell for each line $\ell\in L\setminus I$, and call this cell $c_{\ell}$. Then, we consider the set $Q$ of quadrants defined by the intersections of the lines in $I$, each intersection defining 4 quadrants. If $c_{\ell}$ has size at least 3 for some $\ell\in L\setminus I$, then we charge $\ell$ to one of the quadrant of $Q$ formed by the intersection of two lines of $I$ on the boundary of $c_{\ell}$. If $c_{\ell}$ has size two, then we charge $\ell$ to $c_{\ell}$. Note that, by definition of $c_{\ell}$, no quadrant can be charged more than once. Similarly, no cell of size two can be charged more than once. Hence the number of lines in $L\setminus I$ cannot exceed the sum of the number of quadrants and the number of cells of size two:
\begin{eqnarray*}
|L\setminus I| & \leq & |Q| + \frac n2, \\
n - i & \leq & 4{i\choose 2} + \frac n2, \\
i & \geq & \frac 14 \left(\sqrt{4n + 1} + 1\right) \geq \frac{\sqrt{n}}2 . 
\end{eqnarray*}
\end{proof}

\begin{theorem}\label{th:ubis}
Given a set $L$ of $n$ lines, an independent set of the corresponding $\Hlinecell$ hypergraph has size at most $2n / 3$.
\end{theorem}
\begin{proof}
Let $S$ be an independent set of lines in $\Hlinecell$. This means that, in the corresponding arrangement $\mathcal{A}_{S}$, each cell is touched by at least a line $\ell \in L \setminus S$. Each line $\ell \in L\setminus S$ crosses $|S|+1$ cells of $\mathcal{A}_{S}$. There are ${|S| \cdot (|S|+1) \over 2}+1$ cells in total, and thus

$$|L \setminus S| = n - |S| \ge {|S| \cdot (|S|+1) +2 \over 2 (|S|+1)} $$
$$n \ge {3|S|\over 2}+{1\over |S|+1}$$

\noindent and therefore we conclude that $|S| < {2n/3}$.
\end{proof}

\subsection{Chromatic number of $\Hlinecell$}
In this section, we study the problem of coloring the $\Hlinecell$ hypergraph. That is, we want to color the set $L$ so that no cell is monochromatic. We start by giving an upper bound on the required number of colors.

\begin{theorem}\label{UB:chrom}
Any arrangement of $n$ lines can be colored with at most $2\sqrt{n}+O(1)$ colors so that no edge of the associated $\Hlinecell$ hypergraph is monochromatic.
\end{theorem}
\begin{proof}
Our coloring scheme is as follows: select the largest independent set $I$, color all the lines of $I$ with the same color, remove $I$ from $L$. We now iterate on the remaining lines, where in each step of the algorithm a different color is assigned to the lines we remove. The algorithm stops whenever the number of non-colored lines is at most $n_0$ (the exact value of $n_0$ will be determined afterwards). Whenever $n_0$ or fewer lines remain, we complete the coloring by adding a new color to each of the remaining lines. 

In order to show that this method provides a proper coloring, first observe that any independent set of $L' \subseteq L$ is also an independent set of $L$. That is, any set of lines with the same color assigned form an independent set of $L$. In particular, there cannot exist a cell $c$ in which all lines adjacent to $c$ have the same color assigned.

Let $c(n)$ be the maximum number of colors needed for an arrangement of $n$ lines. We will prove that $c(n)=2\sqrt{n}+O(1)$ using induction. Recall that, by Theorem \ref{th:is}, the size of a maximal independent set is at least $\sqrt{n}/2$. Let $n_0$ be the smallest integer such that $\sqrt{n_0}>1+\sqrt{n_0-\sqrt{n_0}/2}$. Our coloring strategy gives the following recursion for any $n\geq n_0$. 

$$ c(n)\leq c(n-k\sqrt{n})+1 \leq 2\sqrt{n-\sqrt{n}/2}+O(1)+1 \leq 2(\sqrt{n}-1) +O(1)+1 < 2\sqrt{n} +O(1)$$
\end{proof}

We now construct a slightly sublogarithmic lower bound for the chromatic number of $\Hlinecell$:

\begin{theorem}
\label{LB:chroma}
There exists an arrangement of $n$ lines whose corresponding hypergraph $\Hlinecell$  has chromatic number $\Omega(\log n /\log \log n)$.
\end{theorem}
The proof of this claim is constructive. In the following we construct a set of (roughly) $k^k$ lines, in which any $k$-coloring will contain a monochromatic cell (for any $k>0$).  Since we are interested in the asymptotic behavior, it suffices to prove for the case in which $k+1$ is a power of two (that is, $k+1=2^q$ for some $q\geq 0$). In order to proceed with the proof, we first introduce some definitions and helpful results.

For any $x_0\in \mathbb{R}$ we consider the order in which we traverse the lines of $L$ in the vertical line $x=x_0$ from top to bottom. Although the permutation obtained will depend on $L$ and $x_0$, there will be exactly $n\choose 2$ different permutations in any set $L$ of $n$ lines. Let $\mathcal{A}_L$ be the set of different permutations that we can obtain. Each of these permutations is called a {\em snapshot} of $L$.

With this definition we can give an intuitive idea of our construction. Consider any coloring with $k$ colors of a set of $k+1$ lines. By the pigeonhole principle there will be two lines with the same assigned color. Moreover, since the two lines must cross, these two lines must be consecutive in the ordering given by some snapshot. Our approach is to cross these two lines with a second pair of lines with the same color assigned, hence obtaining a monochromatic quadrilateral. The main difficulty of the proof is that the line set must satisfy this property for any $k$-coloring of $L$. In particular, we do not know at which snapshot will the two lines of the same color meet.

We say that a set of snapshots $\mathcal{W} \subseteq \mathcal{A}_L$ is a {\em witness} set of $L$ if, for any two lines $\ell,\ell'\in L$, there exists a snapshot $\pi\in \mathcal{W}$ in which the two lines appear consecutively in $\pi$. 
 It is easy to see that the whole set $\mathcal{A}_L$ is a witness sets of quadratic size for any set of lines. Since the size of the witness set has a direct impact on our bound, we first show how to construct a witness set of smaller size:

\begin{lemma}\label{LB:basicgad}
For any $q\geq 0$ there exists a set $L$ of $2^q$ lines and a witness set $\mathcal{W}$ such that $|\mathcal{W}|\leq  2^{q+1}$.
\end{lemma}
\begin{proof}
We construct the arrangement by induction on $q$ (recall that we assumed $k+1=2^q$ for some $q\geq 0$). For $q=0$ our base gadget $G_1$ simply consists of a single line. Note that the witness property is trivially true, since there don't exist two distinct lines in $G_1$, hence we define $\mathcal{W}_1=\emptyset$.

As we are only interested in the ordering in which lines are crossed, we can do any transformation to a set $L$ of lines, provided that  transformation preserves the permutations in the set $\mathcal{A}_L$. If we update the coordinates of the snapshots in the witness set accordingly, the witness property will still hold. In particular, we can transform a set $L$ of lines so that they become almost parallel and have any desired slope. We call this operation the {\em flattening} of $L$.

With this operation in mind we can do the induction step as follows: for any $q>0$ generate a copy of gadget $G_{2^{q-1}}$ and flatten the lines so that they all have small positive slope and all crossings between lines occur below the horizontal line $y=0$. Let $G'$ be the transformed set of lines and $G''$ be the reflexion of $G'$ with respect to line $y=0$. Gadget $G_{2^q}$ is defined as the union of $G'$ and $G''$ (see Figure \ref{fig_basegadget}).

Observe that $G_{q}$ satisfies the following properties:
\begin{figure}
   \begin{center}
     \includegraphics[width=0.7\textwidth]{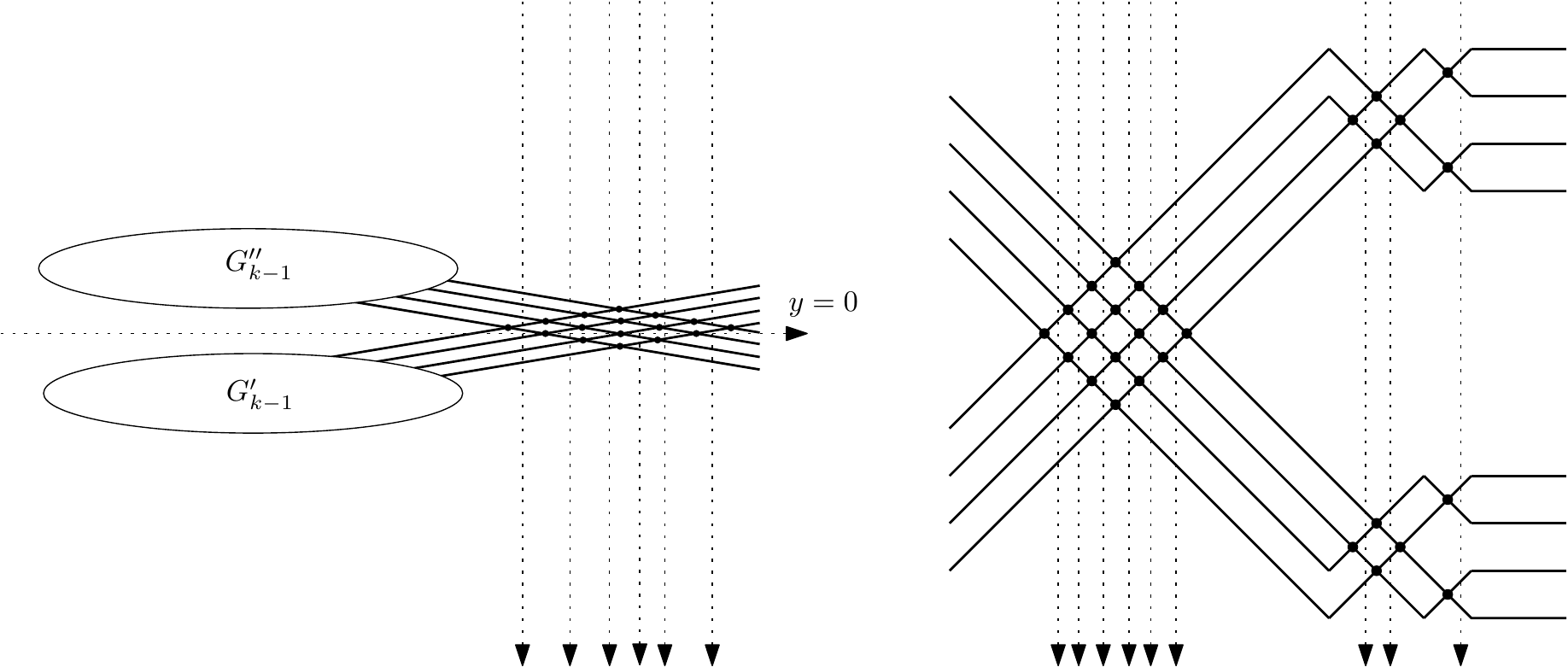}
     \caption{Induction step in the gadget $G_{2^q}$ construction (left). The additional snapshots are also shown  (dashed vertical lines). The generated arrangement for $q=8$ and its witness set is shown in the right. For clarity lines of $G_8$ have been depicted as pseudolines.}
     \label{fig_basegadget}
   \end{center}
\end{figure}

\begin{itemize}
\item[] $(i)$ Gadget $G_{2^q}$ has size exactly $2^q$. Moreover, any two lines cross exactly once.
\item[] $(ii)$ The witness set $\mathcal{W}_{2^{q-1}}$ of $G'$ also acts as witness set of $G''$.
\item[] $(iii)$ The lines of $G'$ and $G''$ intersect in a grid-like fashion, forming cells of size $4$ and $2$. 
\end{itemize}

Observe that property $(i)$ certifies that the construction is a valid set of lines, while properties $(ii)$ and $(iii)$ help us obtain a witness set $\mathcal{W}_{2^q}$ of small size; the crossing between lines of different gadgets  can be guarded with $n-2=2^q-2$ lines (see Figure \ref{fig_basegadget}). Moreover, the crossings between lines of the same gadget can be guarded by the witness set $\mathcal{W}_{2^{q-1}}$ (which by induction satisfies $|\mathcal{W}_{2^{q-1}}|\leq 2\times 2^{q-1} =2^q$). By construction we have that $|\mathcal{W}_{2^q}|= (2^q-2)+ |\mathcal{W}_{2^{q-1}}| \leq 2^q + 2^q =2^{q+1}$.

To finish the proof we must show that $\mathcal{W}_{2^q}$ is indeed a witness set of $L$: Let $\ell,\ell'$ be any two lines of $L$. If these lines belong to the same sub-gadget, we can apply induction and obtain that they must be consecutive in one of the first snapshots. Otherwise, the two lines belong to different sub-gadgets of $G_{2^q}$, hence they will be consecutive at the latter snapshots.
\end{proof}

With the preceding result we can now prove Theorem \ref{LB:chroma}:

\begin{proof}
Let $L^{(0)}$ be the set of lines constructed in Lemma \ref{LB:basicgad} and let $\mathcal{W}$ be the witness set of $L^{(0)}$ (recall that we have $|L^{(0)}|=k+1=2^q$ and $|\mathcal{W}|=m$ for some $m\leq 2(k+1)$). Also, let $\pi_1,\ldots,\pi_m$ be the snapshots of $\mathcal{W}$, sorted from left to right. Consider now any coloring of $L^{(0)}$ with $k$ colors. By the pigeonhole principle, there must exist two lines $\ell,\ell'$ with the same color assigned. Since $\mathcal{W}$ is a witness set, these two lines must be consecutive at some snapshot $\pi \in \mathcal{W}$. Whenever this happens, we say that $\ell$ and $\ell'$ form a {\em monochromatic consecutive pair} at snapshot $\pi$.

In the following, we generalize the above construction to a set $L^{(i)}$ of size $(k+1)^{i+1}$ for any $i \in \{0,\ldots, m\}$. The key property is that in any $k$-coloring of the set $L^{(i)}$, either there is a monochromatic cell or there exist two lines that form a monochromatic consecutive pair at some snapshot $\pi_j$ (for $m\leq j>i$). In particular, notice that the second condition cannot occur for set $L^{(m)}$, hence  there must exist a monochromatic cell in any $k$-coloring of $L^{(m)}$.

We construct the set $L^{(i)}$ by induction on $i$. For $i=0$ the claim is true by Lemma \ref{LB:basicgad}, hence we can focus on the inductive step. For any $i>0$, we construct $L^{(i)}$ with $k+1$ different copies of set $L^{(i-1)}$ flattened so that they satisfy the following properties:

\begin{itemize}
\item $(i)$ For any $j\in \{i,\ldots m\}$, The snapshot of each copy $L^{(i-1)}$ at coordinate $x=j$ is $\pi_j$.
\item $(ii)$ No two lines of the same copy  of $L^{(i-1)}$ cross in the vertical strip $\{-m<x<i\}$. In particular,  the snapshot taken at any coordinate of the strip is $\pi_i$.
\item $(iii)$ Lines of two different copies of $L^{(i-1)}$ cross in the  vertical strip $\{0<x<i\}$ in a grid-like fashion. In particular any two lines that are consecutive in $\pi_i$ form a quadrilateral with other two consecutive lines of another copy of $L^{(i-1)}$.
\end{itemize}
This construction can be done by flattening all the copies of $L^{(i-1)}$ so that each copy essentially becomes a thick line, and placing the different copies in convex position\sholong{.}{(see Figure \ref{fig_induc}). We define the set $\mathcal{W}^{(i)}$ as the set $\{\pi_{i+1}, \ldots, \pi_m\}$ (that is, we remove $\pi_i$ from $\mathcal{W}^{i-1}$).}
 Observe that, since $L^{(i)}$ is composed of $k+1$ different copies of $L^{(i-1)}$, we indeed have $|L^{(i)}|=(k+1)|L^{(i-1)}|=(k+1)^{i+1}$.\sholong{}{Moreover, the size of $\mathcal{W}^{(i)}$ decreases by one in each iteration, hence $|\mathcal{W}^{(i)}|=m-i$.}

In order to complete the proof we must show that, in any coloring $c$ of $L^{(i)}$, we either have a monochromatic cell or a  monochromatic consecutive pair in $\pi_j$ (for some $m\leq j>i$). Apply induction to the different copies of $L^{(i-1)}$: if at least one of the copies has a monochromatic cell or has its monochromatic consecutive pairs at snapshot $\pi_j$ (for some $j>i$) we are done, since the same property will hold for $L^{(i)}$. The other case occurs when all copies of $L^{(i-1)}$ have their monochromatic consecutive pair at snapshot $\pi_i$. Let $\ell_j,\ell'_j$ be the monochromatic consecutive pair of the $j$-th copy of $L^{(i-1)}$ and let $c_j$ be its color. By the pigeonhole principle, there must be two distinct indices $u,v\leq k$ such that $c_u=c_v$. By property $(iii)$ of our construction, the lines $\ell_u,\ell'_u,\ell_v,\ell'_v $ form a quadrilateral in the arrangement of lines of $L^{(i)}$. The quadrilateral will be monochromatic, since by definition the four lines have the same color assigned.
\end{proof}

\subsection{Other coloring results}
For the sake of completeness, we end this section by stating two easy results on coloring vertices or cells instead of lines.

\begin{theorem}\label{tight:chromvertex}
The chromatic number of $\Hvertexcell$ is 2.
\end{theorem}
\noindent Remark that cells of size two only have one vertex, hence cannot be polychromatic. Therefore,  we only consider cells of size at least 3.
\begin{proof}
It is known that the graph obtained from an Euclidean arrangement of lines by taking only the
bounded edges of the arrangement has chromatic number 3 \cite{FHNS}: Sweep the arrangement with a line from left to right. In this ordering, every vertex in the arrangement is adjacent to exactly two predecessors and hence we can assign the colors greedily, such that each vertex has a color different from at least one of its predecessors.
Finally, we can identify two of the three colors, and no cell which is at least a triangle can be monochromatic.
\end{proof}

The following well-known result considers the coloring of cells so that no line is only adjacent to cells of a single color class:

\begin{theorem}[Folk.]\label{tight:chromcells}
The chromatic number of $\Hcellzone$ is 2.
\end{theorem}
\begin{proof}
This claim is equivalent to the fact that the dual graph of the arrangement (where vertices are 
faces, and there is an edge between two faces if they are adjacent) is bipartite. This result
has appeared in recreational texts and concited some research as well \cite{L1894,Gru5}.
\end{proof}

\section{Guarding Arrangements}\label{sec:guarding}
\sholong{}{We now consider the vertex cover problem of the above hypergraphs. That is, we would like to select the minimum number of vertices so that any hyperedge is adjacent to the selected subset. Geometrically speaking, we would like to select the minimum number of vertices (or cells or lines), so that each cell (or line or cell, respectively) contains at least one of the selected items. Recall that this problem is the complementary of the independent set problem. Hence, for each case we will study the easiest of the two problems.   

}
\subsection{Guarding cells with vertices}

We first consider the following problem: given an arrangement of lines $\mathcal{A}$, how many vertices do we need to pick in order to guard the whole arrangement when lines act as obstacles blocking visibility? This can be rephrased as finding the smallest subset of vertices $V$ so that each cell contains a vertex in $V$, and thus we are looking for bounds on the size of a vertex cover for $\Hvertexcell$.

\begin{theorem}\label{LB:guards}
For any set $L$ of $n$ lines, a vertex cover of the corresponding $\Hvertexcell$ hypergraph has size at most ${n^2}/6+2n$. Furthermore, $n^2/6$ vertices might be necessary.
\end{theorem}
\begin{proof}
First notice that any arrangement can have at most $2n$ cells of size exactly $2$. Hence, these cells can be easily guarded with $2n$ guards. Thus, we focus our attention to cells of size $3$ or more. We will guard these cells via a 3-coloring of the vertices of the arrangement. We sweep the arrangement in a fixed direction, and color the vertices in order. 

Observe that the graph of the arrangement is 4-regular, and when a vertex $v$ is encountered, it has exactly two neighbors (say $u,w$) that have already been colored. If $u$ and $w$ have distinct colors, then we assign the third color to $v$. If they have the same color, then we consider the colors assigned to the vertices of the cell having the segments $uv$ and $wv$ on its boundary. If only two colors are present in the cell, we assign the third one to $v$. Otherwise, we assign arbitrarily one of the two possible colors to $v$. With this construction, it is easy to check that all cells of size at least 3 have vertices of three distinct colors on their boundary. In particular, the vertices of any color can guard all cells of size $3$ or more. Since we used three colors and the total number of vertices is $n^2/3$, there will be a color class with at most $n^2/6$ vertices.

For the lower bound we will use the construction of Furedi~\emph{et al.}~\cite{furedi}. This construction creates a family of arrangements that has $n^{2}/3$ triangles in which any vertex of the arrangement is incident to at most two of these triangles. In particular, any vertex cover of the triangles will need at least $n^2/6$ vertices.
\end{proof}

Recall that the hypergraph $\Hvertexcell$ has ${n\choose 2} ={n^2 \over 2} -{n \over 2}$ vertices. Combining this fact with the preceding bounds on the size of a vertex cover allow us to get similar bounds for the independent set problem:  
\begin{corollary}\label{cor:VCIS}
For any set $L$ of $n$ lines, a maximum independent set of the corresponding $\Hvertexcell$ hypergraph has size at least ${n^2}/3-O(n)$. Furthermore, there exists sets of lines whose largest independent set has size at most ${n^2\over 3} - {n\over 2}$.
\end{corollary}

\subsection{Guarding lines with cells}

Here we consider the problem of touching all lines of $L$ with a smallest subset of cells, i.e., we look for bounds on the size of a vertex cover for $\Hcellzone$.


We begin with a simple proof that a minimal vertex cover of $\Hcellzone$ hypergraph has size at most $\lceil{n \over 2}\rceil$, that we will improve below.
\begin{theorem}\label{warmup:guardwithcells}
Given a set $L$ of $n$ lines, a minimal vertex cover of the corresponding $\Hcellzone$ hypergraph has size at most $\lceil{n \over 2}\rceil$.
\end{theorem}
\begin{proof}
We describe a greedy algorithm to find a vertex cover of size $\lceil{n \over 2}\rceil$; we start with an empty set $L$. We find a pair $p,q$ of lines that we still have to cover. Since every two lines cross, there must exist a cell $c$ adjacent to both $p$ and $q$. We add that cell $c$ to the set $L$, and proceed with the unguarded cells. In the last step, if a single line $\ell$ remains to be covered we add  to $L$ any cell touching $\ell$. Since each cell (except the last one) of $L$ guards at least two lines, at most $\lceil{n \over 2}\rceil$ cells will be added into $L$.
\end{proof}

We next provide a lower bound, and improve as well on the upper bound, for large values of $n$. 
\begin{theorem}\label{guardwithcells}
Given any set $L$ of $n$ lines, a minimal vertex cover of the corresponding $\Hcellzone$ hypergraph has size at most ${19n \over 48} +o(n)$. Moreover, there exists a set $L$ of $n$ lines, such that every vertex cover of the corresponding $\Hcellzone$ hypergraph has size at least ${n \over 4}$.
\end{theorem}
\begin{proof}
The lower bound is proved by the fact that there exist arrangements where the largest cell has size $4$ (see \cite{LLMSU07}). This implies that each cell touches at most $4$ lines, and therefore $n/4$ cells are required to touch them all.

The proof of the upper bound claim is a refined version of the method in the preceding theorem: we
first select cells of size four or more and add them to $L$, until any remaining cell that we add to our set is not guaranteed to cover more than three new lines. We then continue adding cells that cover at least three lines in the same fashion. Finally, we complete our construction with cells that cover two lines as in Theorem~\ref{warmup:guardwithcells}. \sholong{Details can be found in the Appendix}{

\proofguardwithcells}.
\end{proof}

\begin{corollary}\label{cor:CZIS}
For any set $L$ of $n$ lines, a maximum independent set of the corresponding $\Hcellzone$ hypergraph has size at least ${n^2 \over 2} + {5n \over 48} -o(1)$ and at most ${n^2 \over 2} + {n \over 4} +1$.
\end{corollary}
This Corollary follows directly from the preceding theorem, the fact that the complement of a vertex cover is an independent set, and that  
any arrangement of $n$ lines in general position has ${n(n+1)\over 2}+1$ cells (hence, the $\Hcellzone$ hypergraph will have that many vertices).

\subsection{Guarding cells with lines}
For the sake of completeness, we also give bounds on the number of lines needed to guard (touch) all cells.

\begin{corollary}\label{cor:CL}
For any set $L$ of $n$ lines, its minimal vertex cover of the corresponding $\Hlinecell$ hypergraph has size at least $n/3$ and at most $n-{\sqrt{n} \over 2}$.
\end{corollary}
Proof of the lower bound is a direct consequence of the complementariness of the vertex cover/independent set and Theorem \ref{th:ubis} (upper bound on the maximum independent set of $\Hlinecell$). Analogously, the upper bound is a consequence of Theorem \ref{th:is}.

\section{Concluding Remarks}\label{sec:conclusion}
Clearly, the main open problems arising from our work consist of closing gaps (when they exist) between lower and upper bounds; this is especially interesting in our opinion for the
problem of coloring lines without producing any monochromatic cell. \sholong{}{We observe that most of our observations hold for pseudo lines as well. Hence, another natural extension would be studying how do the bounds change when we consider families of curves that any two intersect at most $t$ times (for some constant $t>0$).}

However, it is worth noticing that there are several computational issues that are interesting
as well. For example, it is unclear to us which is the complexity of deciding whether a given
arrangement of lines admits a two-coloring in which no cell is monochromatic.

\sholong{\newpage}{}
{

}
\sholong{
\appendix
\section{Proofs omitted from the Document}

\newtheorem*{guardwithcells}{Theorem \ref{guardwithcells}}
\begin{guardwithcells}
Given any set $L$ of $n$ lines, a minimal vertex cover of the corresponding $\Hcellzone$ hypergraph has size at most ${19n \over 48} +o(n)$. Moreover, there exists a set $L$ of $n$ lines, such that every vertex cover of the corresponding $\Hcellzone$ hypergraph has size at least ${n \over 4}$.
\end{guardwithcells}
\begin{proof}
Proof of the lower bound was given in the main document, hence we focus in the upper bound. \proofguardwithcells
\end{proof}
}{}
\end{document}